\documentclass[12pt,a4paper]{article}

\usepackage{amsmath,amsfonts,amssymb,amsthm,mathtools}
\usepackage{fullpage,hyperref,xspace}
\usepackage[footnotesize]{caption}
\usepackage[caption=false,font=footnotesize]{subfig}
\usepackage{graphicx}
\usepackage{xcolor}
\graphicspath{{./figs/}}

\newtheorem{theorem}{Theorem}[section]

\newtheorem{lemma}[theorem]{Lemma}

\newtheorem{definition}[theorem]{Definition}
\newtheorem{remark}[theorem]{Remark}

\DeclareMathOperator{\sign}{sign}
\DeclareMathOperator{\signc}{sign_{\bb C}}
\DeclareMathOperator{\supp}{supp}

\DeclareMathOperator{\diag}{diag}

\DeclareMathOperator*{\argmin}{arg\min}

\newcommand{\bs}{\boldsymbol}
\newcommand{\bb}{\mathbb}
\newcommand{\cl}{\mathcal}

\newcommand{\scp}[2]{\langle #1,\, #2 \rangle}

\newcommand{\inv}[1]{\frac{1}{#1}}

\newcommand{\ts}{\textstyle}

\newcommand{\ie}{\emph{i.e.},\xspace}
\newcommand{\eg}{\emph{e.g.},\xspace}
\newcommand{\im}{{\sf i}}
\newcommand{\Id}{{\rm \bf I}}

\newcommand{\iid}{%
    \ifmmode
        \mathrm{i.i.d.}%
    \else%
        i.i.d.\@\xspace%
    \fi%
}

\newcommand{\st}{%
    \ifmmode
        \quad\mathrm{s.t.}\quad%
    \else%
        s.t.\@\xspace%
    \fi%
}

\newcommand{\whp}{\mbox{w.h.p.\@}\xspace}

\title{The importance of phase\\ in complex compressive sensing}
\author{Laurent Jacques$^*$ and Thomas Feuillen\thanks{LJ and TF are with the ICTEAM institute in UCLouvain,
    Belgium. E-mail: \url{{laurent.jacques,thomas.feuillen}@uclouvain.be}. LJ is funded by the Belgian F.R.S.-FNRS.}}

\begin{document}
\maketitle

\begin{abstract}
We consider the question of estimating a real low-complexity signal (such as a sparse vector or a low-rank matrix) from the phase of complex random measurements. We show that in this \emph{phase-only compressive sensing} (PO-CS) scenario, we can perfectly recover such a signal with high probability and up to global unknown amplitude if the sensing matrix is a complex Gaussian random matrix and if the number of measurements is large compared to the complexity level of the signal space. Our approach proceeds by recasting the (non-linear) PO-CS scheme as a linear compressive sensing model built from a signal normalization constraint, and a phase-consistency constraint imposing any signal estimate to match the observed phases in the measurement domain. Practically, stable and robust signal direction estimation is achieved from any \emph{instance optimal} algorithm of the compressive sensing literature (such as basis pursuit denoising). This is ensured by proving that the matrix associated with this equivalent linear model satisfies with high probability the restricted isometry property under the above condition on the number of measurements. We finally observe experimentally that robust signal direction recovery is reached at about twice the number of measurements needed for signal recovery in compressive sensing.        
\end{abstract}  

\section{Introduction}
\label{sec:introduction}

About 40 years ago, Oppenheim and his collaborators~\cite{Oppenheim_1982,Oppenheim_1981} determined that, under certain conditions and up to a global unknown amplitude, a band-limited signal can be exactly recovered from the phase of its Fourier transform; one can reconstruct the \emph{direction} of this signal. The authors also provided a practical iterative algorithm for this purpose. Since the observed signal lies at the intersection of the set of band-limited signals and a specific phase-consistency set --- the set of signals matching, in the frequency space, the phase of the observed signal --- this algorithm is designed to find this intersection, when it is unique, by alternate projections onto these two convex sets (the POCS method). 

More recently, Boufounos~\cite{boufounos2013sparse} considered a similar question in the context of complex compressive sensing where measurements of a sparse signal are obtained from its multiplication with a complex, fat sensing matrix. Since keeping the phase of such complex measurements naturally extends the formalism of one-bit compressive sensing (which keeps only the sign of real random projections)~\cite{boufounos2008,jacques2013robust,plan2012robust}, the author demonstrated that, in the case of complex Gaussian random sensing matrices, one can estimate the direction of such a signal from the phase of the compressive measurements (or \emph{phase-only} compressive sensing -- PO-CS). The estimation error then provably decays when the number of measurements increases. Boufounos actually proved that, with high probability (\whp) and up to a controlled distortion, the discrepancy between the phase of complex Gaussian random projections of two sparse vectors encodes their Euclidean distance. Recently, pursuing the correspondence with one-bit CS and extending its central $(\ell_1,\ell_2)$-restricted isometry property ($(\ell_1,\ell_2)$-RIP)~\cite{foucart2016flavors} to complex sensing matrices, Feuillen and co-authors established that a single-step procedure called projected back-projection (PBP) --- an estimation of the signal direction by hard-thresholding the multiplication of the PO-CS measurements by the adjoint sensing matrix~\cite{FDVJ19}  --- also yields small and decaying reconstruction error. Of interest for our work, Boufounos also designed in~\cite{boufounos2008} a greedy algorithm delivering an estimate that is both sparse and phase-consistent with the unknown signal; this signal and the estimate share identical phases in the random projection domain. However, none of the theoretical approaches above can explain why this specific algorithm succeeds numerically in perfectly estimating the observed signal direction. 

With this work, we bring two main contributions to this context. First, we demonstrate that, in a noiseless scenario, perfect estimation of a signal direction from the phase of its random, complex Gaussian measurements is possible, not only for sparse signals but for any signals belonging to a symmetric, low-complexity conic set (or \emph{cone}) of reduced dimensionality, including (union of) subspaces, the set of sparse signals, the set of model-based or group sparse vectors~\cite{baraniuk2010model,Ayaz_2016} or the set of low-rank matrices~\cite{fazel2002matrix}. This result strikingly differs from known reconstruction guarantees in the context of (real) one-bit CS~\cite{boufounos2008,jacques2013robust,plan2012robust}; in this case, the direction of a low-complexity signal can only be estimated up to a lower-bounded error~\cite{jacques2013robust}.  

Second, we show that one can reconstruct the direction of a low-complexity signal from any \emph{instance optimal} algorithms of the CS literature whose error guarantees are controlled by the restricted isometry property of the sensing matrix~\cite{candes2005decoding,Foucart_2013,Traonmilin_2018}.  Using such algorithms, we can bound the reconstruction error of the direction of a low-complexity signal observed in a PO-CS model. This error bound is \emph{(i)} \emph{non-uniform}, in the sense that the estimation is possible, \whp, given the observed low-complexity signal, and \emph{(ii)} \emph{stable and robust} as the instance optimality of the selected algorithm allows for both a bounded noise on the observed measurement phases and an error on the modeling of the signal by an element of a low-complexity set. Experimentally, we observe that the number of measurements required for robust estimation of signal direction from noiseless phase-only observation is about twice the one needed for signal estimation in the case of (linear)~CS. 

Our approach is thus complementary to~\cite{FDVJ19}; in this work, provided that the sensing matrix respects the $(\ell_1,\ell_2)$-RIP, the PBP algorithm yields a bounded reconstruction error for the estimation of the direction of any unknown signal in (noisy) PO-CS. The bound thus holds \emph{uniformly}. However, the observed signal must be sparse and the reconstruction error does not vanish in absence of measurement noise.  

Regarding the mathematical tools, all our developments rely on recasting of the PO-CS recovery problem as a linear compressive sensing problem. This one is built from a signal normalization constraint --- as allowed from the independence of the PO-CS model to the signal amplitude --- and a phase-consistency constraint imposing to match the observed measurement phases. We show that in noiseless PO-CS, provided that the number of measurements is large compared to the signal set complexity, the sensing matrix associated with this equivalent formulation respects the restricted isometry property over the signal set; the distance between two low-complexity vectors is encoded in the one of their projections by this matrix (see Sec.~\ref{sec:sensing-model}). Moreover, when a bounded noise corrupts the PO-CS model, this property is preserved if the noise variations are small.  All our proofs are self-contained and mostly rely on extending to the complex field previously introduced random embeddings, such as the (local) sign-product embedding~\cite{plan2012robust,jacques2013quantized,foucart2016flavors}.

The rest of this paper is structured as follows. We first end this introduction by defining some useful notations and conventions. In Sec.~\ref{sec:sensing-model}, we recall a few fundamental results of the compressive sensing theory that are subsequently used in this work. Sec.~\ref{sec:PO-CS} defines the PO-CS problem, its reformulation as a signal recovery problem observed by a linear model and provides the first guarantees on the exact estimation of a low-complexity signal in the case of noiseless PO-CS. Sec.~\ref{sec:noisy-phase-sensing} extends these guarantees to the case where the PO-CS model is corrupted by a bounded noise and the unknown signal is well represented in a low-complexity domain with a bounded modeling error.  For clarity, the proofs of our main results are inserted in Sec.~\ref{sec:proofs}. In Sec.~\ref{sec:experiments}, we provide several experiments leveraging the instance optimality of the basis pursuit denoising algorithm. We first establish from which number of measurements one can perfectly estimate the direction of a sparse signal in PO-CS with complex Gaussian random sensing matrices. We then measure the accuracy of such an estimate under noise corruption for various noise levels. We conclude this paper in Sec.~\ref{sec:conclusion} with a few open questions and perspectives. 

\paragraph*{Notations and conventions:} Hereafter, the symbols $C,c > 0$ represent absolute constants whose exact value may change from one instance to the other.
We denote matrices and vectors with bold symbols, \eg $\bs \Phi \in \mathbb{C}^{m\times n}$, $\bs x \in \mathbb{C}^{n}$, and scalar values with light symbols. We will often use
the following quantities: $[d]:=\{1,\,\cdots, d\}$ for $d \in \bb N$; $\im = \sqrt{-1}$; $\Re\{\lambda\}$ (or $\lambda^\Re$) and $\Im\{\lambda\}$ (or $\lambda^\Im$) are the real and
imaginary part of $\lambda \in \bb C$, respectively, and
$\lambda^\ast$ is its complex conjugate; $\bs 1_d := (1,\,\cdots, 1)^\top \in \bb R^d$ and $\Id_d$ is the $d \times d$ identity matrix; $\bs A^*$ is the adjoint (conjugate transpose) of $\bs A$; $\supp \bs x = \{i: x_i \neq 0\}$; $|\cl S|$ is the cardinality of a finite set $\cl S$; $\bs A_\Omega$ is matrix formed by restricting the $d$ columns of $\bs A$ to those indexed in $\Omega \subset [d]$; $\langle\bs x, \bs y  \rangle= \sum_{i=1}^d x^*_i  y_i$ is the scalar product of $\bs x, \bs y \in \mathbb{C}^d$; $\| \bs x \|_p=(\sum_{i=1}^d |x_i|^p)^{1/p}$ is the $\ell_p$-norm of $\bs x$ ($p\geq 1$), with $\|\bs x\|_\infty = \max_i |x_i|$ and $\|\bs x \|_0:= |\supp (\bs x) |$; $\cl K - \cl K' := \{\bs u - \bs v: \bs u \in \cl K,\bs v \in \cl K'\}$ and $\bb R \bs x$ denote the Minkowski difference of $\cl K, \cl K' \subset \bb R^d$ and the span of $\bs x \in \bb R^d$, respectively
. The symbol ``$\sim$'' is used to describe either the distribution law of a random quantity (\eg $X \sim \cl N(0,1)$), or the equivalence of distribution between two random quantities (\eg $Y \sim X$). Given a (scalar) random distribution $\cl D$, we denote by $\cl D^{m \times n}$ the $m \times n$ random matrix distribution generating matrices with entries independently and identically distributed ($\iid$) as $\cl D$, and we omit the superscript $n$ for vectors (for which $n=1$); for instance, the $m \times n$ real and complex Gaussian random matrices distributed as $\cl N^{m\times n}(\mu', \sigma^2)$ and $\cl N_{\bb C}^{m\times n}(\mu, 2\sigma^2) \sim \cl N^{m\times n}(\mu^\Re, \sigma^2) + \im \cl N^{m\times n}(\mu^\Im, \sigma^2)$, respectively, for some means $\mu,\mu'$ and variance $\sigma^2$. An $s$-sparse $\bs x$ vector belongs to the set $\Sigma^n_s :=\{\bs u \in \bb R^n, \|\bs u \|_0 \leq s \}$. 

\section{Preliminaries}
\label{sec:sensing-model}

Let us first provide the key principles of CS theory. In the complex field, the compressive sensing of a signal $\bs x \in \bb R^n$ with a complex sensing matrix $\bs A = \bs A^\Re + \im \bs A^\Im \in \bb C^{m \times n}$ corresponds to acquiring $m$ noisy complex measurements through the model~\cite{Foucart_2013}
\begin{equation}
  \label{eq:cs-model}
  \bs y = \bs A \bs x + \bs \epsilon,
\end{equation}
for some complex measurement noise $\bs \epsilon \in \bb C^m$. In this model, we assume $\bs x$ real as we can always recast the sensing of $\bs x \in \bb C^n$ as the one of $(\Re(\bs x)^\top, \Im(\bs x)^\top)^\top \in \bb R^{2n}$ with the sensing matrix $(\bs A^\Re, -\bs A^\Im) + \im (\bs A^\Im, \bs A^\Re) \in \bb C^{m \times 2n}$. 

The signal $\bs x$ can be recovered from the measurement vector $\bs y$, if $\bs A$ is ``well conditioned'' with respect the sparse signal set $\Sigma^n_s$~\cite{candes2005decoding, Foucart_2013}. This happens for instance if, for some $0<\delta<1$, $\bs A$ satisfies the restricted isometry property of order $2s$, or RIP$(\Sigma^n_{2s},\delta)$, defined by
\begin{equation}
\label{eq:RIP-def}
\ts   (1-\delta) \|\bs u\|^2 \leq \|\bs A \bs u\|^2 \leq (1+\delta) \|\bs u\|^2, \quad \forall \bs u \in \Sigma^n_{2s}.
\end{equation}
For instance, this property holds with high probability if $\bs A$ is a real (or complex~\cite{Foucart_2013}) $m \times n$ random Gaussian matrix, and if $m \geq C \delta^{-2} s \log ({n}/{\delta s})$~\cite{baraniuk2008simple}.   

If the RIP is verified over $\Sigma^n_{2s}$ with $\delta < 1/\sqrt 2$~\cite{cai2013sparse} (see also~\cite[Thm 6]{foucart2016flavors}), then, for some $C,D >0$ and $\varepsilon \geq \|\bs \epsilon\|$, the \emph{basis pursuit denoising}~\cite{chen2001atomic} estimate $\hat{\bs x} = \Delta_{\Sigma^n_s}(\bs y, \bs A; \varepsilon)$ with
\begin{equation}
  \label{eq:BP-def}
\ts  \Delta_{\Sigma^n_s}(\bs y, \bs A; \varepsilon) := \argmin_{\bs u} \|\bs u\|_1 \st \|\bs A \bs u - \bs y\| \leq \varepsilon, \tag{BPDN}
\end{equation}
satisfies 
\begin{equation}
  \label{eq:l2-l1-inst-opt}
\ts \|\bs x - \hat{\bs x}\| \leq C s^{-1} \|\bs x - \bs x_s\|_1 + D \varepsilon,
\end{equation}
with $\bs x_s$ the closest $s$-sparse signal to $\bs x$ (that is, its best $s$-sparse approximation). Consequently, if $\bs x \in \Sigma^n_s$, $\bs \epsilon=\bs 0$, and $\varepsilon=0$, we get perfect signal recovery ($\hat{\bs x} = \bs x$).
Similar robustness and stability results can be achieved for a series of other algorithms --- such as orthogonal matching pursuit (OMP), compressive sampling matching pursuit (CoSaMP), or iterative hard thresholding (IHT) --- as soon as $\bs A$ respects the RIP for sparsity levels equal to a few multiples of $s$ and for small enough $\delta$~\cite{Tropp_2007,blumensath2009iterative,foucart2016flavors,needell2009cosamp}
\medskip

In this paper, we are interested in more general low-complexity signal spaces than the set of sparse vectors. We open our study, for instance, to (union of) subspaces of $\bb R^n$, such as band-limited or sparse signals, or signals displaying more complex sparsity patterns such as model-based or group sparsity~\cite{Ayaz_2016,baraniuk2010model}, and to the set $\Lambda^n_r$ of $\sqrt n \times \sqrt n$ rank-$r$ matrices (for $n$ a squared integer). We thus focus on \emph{symmetric cones} $\cl K$ (for which $-\cl K = \cl K$ and $\lambda \cl K \subset \cl K$ for all $\lambda >0$) that occupy a ``small volume'' of $\bb R^n$. This can be measured by the (localized) Gaussian mean width $w(\cl K \cap \bb S^{n-1})$ of $\cl K$, with 
\begin{equation}
  \label{eq:GMW-def}
  w(\cl K') := \bb E \sup_{\bs u \in \cl K'} \scp{\bs g}{\bs u},\ \text{for $\cl K' \subset \bb R^n$ and $\bs g \sim \cl N^n(0,1)$.}  
\end{equation}
This quantity encodes the intrinsic dimension of $\cl K$; for instance, $w^2(\cl L \cap \bb S^{n-1}) \leq C L$ for a subspace $\cl L \subset \bb R^n$ of dimension $L$ (such as the space of band-limited signals of $\bb R^n$ whose discrete cosine transform (DCT) is zero over the last $n-L$ frequencies),  $w^2(\Sigma^n_s \cap \bb S^{n-1}) \leq C s \log(n/s)$ --- that is the quantity driving the number of measurements for sparse signal recovery in CS ---, and $w(\Lambda^n_r \cap \bb S^{n-1}) \leq C r \sqrt n$~\cite{chandrasekaran2012convex,plan2012robust,candes2011tight} (see also~\cite[Table 1]{Jacques_2017} for other examples of low-complexity sets). 

It has now been established that, for a variety of low-complexity, symmetric cones $\cl K$, a signal $\bs x \in \cl K$ can be recovered from the noisy measurement vector $\bs y$ in \eqref{eq:cs-model}. For each of these sets, there exists an algorithm $\Delta_{\cl K}:\bb C^m \times \bb C^{m\times n} \times \bb R_+ \to \bb R^n$ that provides an \emph{instance optimal} estimate $\hat{ \bs x}=\Delta_{\cl K}(\bs y, \bs A; \varepsilon)$ of $\bs x$ verifying, for some $C,D >0$, 
\begin{equation}
  \label{eq:l2-l1-inst-opt-gen}
\ts \|\bs x - \hat{\bs x}\| \leq C e_0(\bs x, \cl K) + D \varepsilon,
\end{equation}
where $e_0$, which vanishes if $\bs x \in \cl K$, stands for a \emph{modeling error} in the approximation of $\bs x$ by an element of $\cl K$. For instance, we have $e_0(\bs x, \Sigma^n_s) = \|\bs x - \bs x_s\|_1/\sqrt s$ or $e_0(\bs x, \Sigma^n_s) = \|\bs x - \bs x_s\|_2$~\cite{Foucart_2013,Traonmilin_2018}, and $e_0(\bs X, \Lambda^n_r) = \|\bs X - \bs X_r\|_*/\sqrt r$ with ${\|\cdot \|_*}$ the nuclear norm and $\bs X_r$ is the best rank-$r$ approximation of $\bs X \in \bb R^{\sqrt n \times \sqrt n}$~\cite{fazel2008compressed}.

This instance optimality holds if, for some $0<\delta_0<1$ associated with $\Delta_{\cl K}$ and $\cl K$ (for instance, $\delta_0 = 1/\sqrt 2$ for \ref{eq:BP-def} and $\cl K = \Sigma^n_s$) and $0<\delta<\delta_0$, $\bs A$ respects the general RIP$(\cl K - \cl K, \delta)$~\cite{Traonmilin_2018} defined by  
\begin{equation}
\label{eq:gen-RIP-def}
\ts   (1-\delta) \|\bs u\|^2 \leq \|\bs A \bs u\|^2 \leq (1+\delta) \|\bs u\|^2, \quad \forall \bs u \in \cl K - \cl K,
\end{equation}
where $\cl K - \cl K$ is the Minkowski difference of $\cl K$ with itself (\eg $\cl K - \cl K = \Sigma^n_{2s}$ for $\cl K = \Sigma^n_s$). 

We introduce the concept of \emph{recoverable} set to capture these last considerations.     
\begin{definition}
\label{def:recoverable-set}
A symmetric cone $\cl K \subset \bb R^n$ is called $(\Delta_{\cl K},\delta_0)$-\emph{recoverable} if there exists an algorithm $\Delta_{\cl K}:\bb C^m \times \bb C^{m\times n} \times \bb R_+ \to \bb R^n$ such that, for some $0<\delta_0 < 1$, any matrix $\bs A \in \bb C^{m \times n}$ satisfying the RIP$(\cl K- \cl K, \delta)$ with $0<\delta<\delta_0$, and any noise $\bs \epsilon \in \bb C^m$ with $\|\bs \epsilon\| \leq \varepsilon$, the vector $\hat{\bs x} = \Delta_{\cl K}(\bs A \bs x + \bs \epsilon, \bs A; \varepsilon) \in \bb R^n$ estimates any signal $\bs x \in \bb R^n$ with bounded modeling error $e_0(\bs x, \cl K)$ according to the instance optimality condition \eqref{eq:l2-l1-inst-opt-gen}. 
\end{definition}

To give a few examples of sets $\cl K$ and algorithms, we can mention that a mixed norm penalization of \ref{eq:BP-def} is adapted for signals displaying structured sparsity~\cite{Ayaz_2016}, the recovery of low-rank matrices is ensured either by replacing the $\ell_1$-norm of the \ref{eq:BP-def} program by the nuclear norm~\cite{fazel2002matrix,fazel2008compressed}, or by resorting to specific iterative algorithms such as Singular Value Projection (SVP)~\cite{jain2010guaranteed} or conjugate gradient iterative hard thresholding (CGIHT)~\cite{Blanchard_2015}. 

Concerning the existence of matrices satisfying \eqref{eq:gen-RIP-def}, Mendelson and Pajor~\cite{Mendelson_2008} proved the following crucial result that we will use in our proofs.
\begin{theorem}[{Adapted from\footnote{Thm 2.1 in~\cite{Mendelson_2008} is actually valid for any sub-Gaussian random matrices. It is also restricted to subsets of $\bb S^{n-1}$, which can be easily extended to cones as in Thm~\ref{thm:gen-rip}.}~\cite[Thm 2.1]{Mendelson_2008}}]
\label{thm:gen-rip}
Given a cone $\cl K' \subset \bb R^n$ and $\delta > 0$, if
\begin{equation}
  \label{eq:medelson-sample-complex}
  m \geq C \delta^{-2} w^2(\cl K' \cap \bb S^{n-1}),  
\end{equation}
then  the matrix $\frac{1}{\sqrt m} \bs \Phi$ with $\bs \Phi \sim \cl N^{m \times n}(0,1)$ respects the RIP$(\cl K', \delta)$ with probability exceeding $1 - C\exp(-c \delta^2 m)$.
\end{theorem}

This result, which extends easily to complex Gaussian random matrices\footnote{Simply observe that for $\bs \Phi \sim \cl N_{\bb C}(0,2)$ and real $\bs x$, $\|\bs \Phi \bs x\|^2 = \| \bar{\bs \Phi} \bs x\|^2$ with $\bar{\bs \Phi} := [ (\bs \Phi^\Re)^\top, (\bs \Phi^\Im)^\top ]^\top \sim \cl N^{2m \times n}(0,1)$.}, shows that, \whp, we can stably recover all signals of a low-complexity cone $\cl K$ from $m$ Gaussian random measurements provided $m \geq C \delta^{-2} w^2((\cl K - \cl K) \cap \bb S^{n-1})$. For similar guarantees in the context of other (structured) sensing matrices (\eg partial Fourier measurements) and specific low-complexity spaces, the interested reader may consult, for instance,~\cite{Foucart_2013} for a comprehensive overview of the literature. 

\section{Phase-only Compressive Sensing}
\label{sec:PO-CS}

We now depart from the linear CS model and consider the possibility of (partly) recovering a signal $\bs x \in \bb R^n$ from the \emph{phase} of random projections achieved with a complex sensing matrix $\bs A \in \bb C^{m \times n}$. For clarity, we first consider the case where $\bs x$ belongs to a low-complexity set $\cl K \subset \bb R^n$ and the phase measurements are exact. We thus study the non-linear, noiseless phase-only compressive sensing (PO-CS) model 
\begin{equation}
  \label{eq:PO-CS}
  \bs z = \signc (\bs A \bs x),
\end{equation}
with $\signc \lambda := \lambda /|\lambda|$ if $\lambda \in \bb C \setminus \{0\}$, and $\signc 0 := 0$. We develop in the next section a noisy variant of \eqref{eq:PO-CS}, and cover the case where $\bs x$ is only known to be close to $\cl K$ (which induces a small modeling error $e_0$ in \eqref{eq:l2-l1-inst-opt-gen}).

This model encompasses the one considered by Oppenheim in~\cite{Oppenheim_1982} when $\cl K$ is the subspace of band-limited signals and $\bs A$ is the Fourier sensing matrix. Moreover, since $\signc \lambda = \sign \lambda$ for $\lambda \in \bb R$, PO-CS naturally generalizes one-bit CS to the complex field~\cite{boufounos2008,boufounos2013sparse}.

We remove the ambiguity raised by the unobserved signal amplitude in \eqref{eq:PO-CS} by arbitrary specifying 
\begin{equation}
  \label{eq:norm-hyp}
  \|\bs A  \bs x\|_1=\kappa \sqrt m,
\end{equation}
for some $\kappa > 0$. This $\ell_1$-norm normalization of $\bs \Phi \bs x$ is reminiscent of the one imposed in one-bit CS where the signal amplitude is unknown~\cite{plan2012robust}. It can also be understood by the fact that, for $\|\bs x\|=1$, and $\bs A = \bs \Phi/\sqrt m$ with $\bs \Phi \sim \cl N^{m \times n}_{\bb C}(0,2)$, $\bb E \|\bs A \bs x\|_1=\kappa \sqrt m$  with $\kappa := \sqrt{\pi/2}$ (see the proof of Lemma~\ref{lem:concentration-L1Phi-norm}). Imposing \eqref{eq:norm-hyp} is thus close to enforcing $\|\bs x\| = 1$ for complex Gaussian random matrices.

Following~\cite{boufounos2013sparse}, this normalization allows us to recast the PO-CS model as a linear CS model. Indeed, introducing $\bs \alpha_{\bs z} := \bs A^* \bs z/(\kappa \sqrt m) \in \bb C^m$, we see that the \emph{phase-consistency} and normalization constraints defined by, respectively,
\begin{equation}
  \label{eq:consistency-root}
  \diag(\bs z)^* \bs A \bs u \in \bb R^m_+,\quad \scp{\bs \alpha_{\bs z}}{\bs u} = 1,
\end{equation}
are respected for $\bs u = \bs x$. Since $\bs x$ is real, this is equivalent to
\begin{equation}
  \label{eq:consistency-real-imag}
  \begin{cases}
        \bs H_{\bs z} \bs u = \bs 0,\quad \scp{\bs \alpha^\Re_{\bs z}}{\bs u} = 1,\quad \scp{\bs \alpha^\Im_{\bs z}}{\bs u} = 0\\
    (\bs D_{\bs z}^\Re \bs A^\Re + \bs D_{\bs z}^\Im \bs A^\Im ) \bs u > \bs 0,
  \end{cases}
\end{equation}
with $\bs u \in \bb R^n$, and $\bs D_{\bs v} := \diag(\bs v)$, $\bs H_{\bs v} := \Im(\bs D_{\bs v}^* \bs A) = \bs D_{\bs v}^\Re \bs A^\Im - \bs D_{\bs v}^\Im \bs A^\Re $ for $\bs v \in \bb C^m$.

A meaningful estimate $\hat{\bs x} \in \bb R^n$ of $\bs x$ should thus respect the constraints \eqref{eq:consistency-real-imag}.
Interestingly, if we discard the second line of \eqref{eq:consistency-real-imag}, the remaining constraints amount to imposing
\begin{equation}
  \label{eq:equiv-cs-model}
\bs A_{\bs z} \bs u = \bs e_1 = (1, 0,\, \cdots, 0)^\top = \bs A_{\bs z} \bs x,  
\end{equation}
with 
\begin{equation}
  \label{eq:Ax-def} 
  \bs A_{\bs v} :=
  \begin{pmatrix}
    (\bs \alpha^\Re_{\bs v})^\top\\[1mm]
    (\bs \alpha^\Im_{\bs v})^\top\\[1mm]
    \bs H_{\bs v}
  \end{pmatrix}
  =
  \begin{pmatrix}
    \inv{\kappa \sqrt m} \big((\bs v^\Re)^\top \bs A^\Re + (\bs v^\Im)^\top \bs A^\Im \big)\\[1mm]
    \inv{\kappa \sqrt m} \big((\bs v^\Re)^\top \bs A^\Im - (\bs v^\Im)^\top \bs A^\Re \big)\\[1mm]
    \bs D_{\bs v}^\Re \bs A^\Im - \bs D_{\bs v}^\Im \bs A^\Re
  \end{pmatrix} \in \bb R^{(m+2)\times n},\quad \bs v \in \bb C^m,
\end{equation}
and where the equality $\bs A_{\bs z} \bs x = \bs e_1$ reformulates the hypothesis $\|\bs A \bs x\|_1 = \kappa \sqrt m$ since $\bs H_{\bs z} \bs x = \bs 0$.

In other words, a subset of the constraints in \eqref{eq:consistency-real-imag} leads to the equivalent CS recovery problem \eqref{eq:equiv-cs-model} where we aim at estimating $\bs x$ from $\bs A_{\bs z} \bs x = \bs e_1$. Therefore,  if one can show that, for some $0<\delta<1$, $\bs A_{\bs z}$ respects \whp the RIP$(\cl K - \cl K, \delta)$ defined in \eqref{eq:gen-RIP-def}, then, any estimate $\hat{\bs x} \in \cl K$ satisfying 
\eqref{eq:consistency-real-imag} is sure to recover the direction of $\bs x$ since
$$
\ts \|\hat{\bs x} - \bs x\|^2 \leq \frac{1}{1-\delta} \|\bs A_{\bs z} \hat{\bs x} - \bs A_{\bs z} \bs x\| = 0.
$$

Moreover, as imposing both \eqref{eq:equiv-cs-model}  and $\hat{\bs x} \in \cl K$ often leads to impractical algorithms (in fact, an NP-hard, constrained $\ell_0$ minimization for the recovery of sparse signals), a more practical result derives easily from the model~\eqref{eq:equiv-cs-model} and from the definition of recoverable set introduced in Sec.~\ref{sec:sensing-model}.

\begin{theorem}[Perfect signal direction estimation in noiseless PO-CS]  
\label{thm:sig-dir-rec-noiseless-po-cs}
Let $\cl K \subset \bb R^n$ be a symmetric cone that is $(\Delta_{\cl K},\delta_0)$-recoverable with the instance optimal algorithm $\Delta_{\cl K}$. Given $\bs A \in \bb C^{m \times n}$, we can recover the direction of $\bs x \in \cl K$ from its $m$ phase-only measurements $\bs z = \signc(\bs A \bs x)$ if the matrix $\bs A_{\bs z} \in \bb R^{(m+2)\times n}$ built from $\bs A$ and $\bs z$ in \eqref{eq:Ax-def} respects the RIP$(\cl K - \cl K, \delta)$ with $0<\delta < \delta_0$. More precisely, under the hypothesis $\bs A_{\bs z} \bs x = \bs e_1$ (which corresponds to $\|\bs A \bs x\|_1 = \kappa \sqrt m$), we have $\hat{\bs x} = \bs x$ if $\hat{\bs x} = \Delta_{\cl K}(\bs e_1, \bs A_{\bs z}; 0)$.  
\end{theorem}
\begin{proof}
The proof is a direct translation of Def.~\ref{def:recoverable-set}; in \eqref{eq:l2-l1-inst-opt}, $e_0(\bs x, \cl K) = 0$ since $\bs x \in \cl K$, and $\varepsilon = 0$ according to the equivalent noiseless sensing model \eqref{eq:equiv-cs-model}.    
\end{proof}

\begin{remark}
We can question the specific choice of the normalization \eqref{eq:norm-hyp} as a way to complement the constraint $\bs H_{\bs z} \bs u = \bs 0$ (and thus $\bs H_{\bs z}$) to reach a RIP matrix $\bs A_{\bs z}$. For instance, we could replace $\bs \alpha_{\bs z}$ in \eqref{eq:Ax-def} by an independent real Gaussian random vector $\bs g$ and remove the signal normalization ambiguity by arbitrarily assuming $\scp{\bs g}{\bs x} = 1$ (which should happen with probability 1). In this case, $\bs A'_{\bs z} := (\bs g^\top, \bs H_{\bs z}^\top)^\top$ would offer another equivalent linear sensing model $\bs A'_{\bs z} \bs u = \bs e_1$, hoping then to prove that $\bs A'_{\bs z}$ respects the RIP. However, the use of \eqref{eq:norm-hyp} is critical. As developed in the proof of the next theorem (see Sec.~\ref{sec:proof-rip-Ax}), for any low-complexity vector $\bs u$ and in the case where $\bs A = \bs \Phi / \sqrt m$ with $\bs \Phi \sim \cl N_{\bb C}^{m \times n}(0,1)$, $\scp{\bs \alpha^\Re_{\bs z}}{\bs u}$ in the first component of $\bs A_{\bs z} \bs u$ estimates the projection $\scp{\frac{\bs x}{\|\bs x\|}}{\bs u}$. Moreover, since $\bs H_{\bs z} \bs x = \bs 0$, $\|\bs H_{\bs z} \bs u\|^2$ approximates $\|\bs u - \scp{\frac{\bs x}{\|\bs x\|}}{\bs u} \frac{\bs x}{\|\bs x\|}\|^2$. By Pythagoras, we can then prove that, \whp, $\|\bs A_{\bs z} \bs u\|^2 = (\scp{\bs \alpha^\Re_{\bs z}}{\bs u})^2 + \|\bs H_{\bs z} \bs u\|^2$ encodes $\|\bs u\|^2$ up to a controlled multiplicative distortion. 
\end{remark}

The key result of this work (proved in Sec.~\ref{sec:proofs}) consists in showing that the matrix $\bs A_{\bs z}$ in \eqref{eq:Ax-def} built from the association of a complex Gaussian random matrix $\bs A$ and a signal $\bs x \in \bb R^n$, respects, \whp, the RIP.   
\begin{theorem}
  \label{thm:rip-for-Ax}
  Given a symmetric cone $\cl K \subset \bb R^n$, $\delta > 0$, $\bs A = \bs \Phi/\sqrt m$ with $\bs \Phi \sim \cl N_{\bb C}^{m \times n}(0,2)$, $\bs x \in \bb R^n$, and $\bs A_{\bs z}$ defined in \eqref{eq:Ax-def} from $\bs A$ and $\bs z = \signc(\bs A \bs x)$ with $\kappa = \sqrt{\pi/2}$, if
  \begin{equation}
    \label{eq:sample-complexity}
    \ts m \geq C (1+\delta^{-2}) w^2\big((\cl K - \bb R \bs x) \cap \bb S^{n-1}\big),
  \end{equation}
  then, with probability exceeding $1 - C \exp(-c \delta^2 m)$, $\bs A_{\bs z}$ satisfies the RIP$(\cl K,\delta)$. 
\end{theorem}

Therefore, if $\bs x \in \cl K$ and $\cl K$ is symmetric (as for the sets of sparse signals and low-rank matrices, or for any union of subspaces), we have $\cl K - \cl K - \bb R \bs x \subset \cl K^{(3)} := \cl K + \cl K + \cl K$, and  
we deduce that, from Thm~\ref{thm:sig-dir-rec-noiseless-po-cs} and Thm~\ref{thm:rip-for-Ax}, we can exactly recover \whp the direction of a signal $\bs x \in \cl K$ from the phase of its complex Gaussian random measurements provided $m$ is large compared to the intrinsic complexity of $\cl K^{(3)}$ --- as measured by $w^2(\cl K^{(3)} \cap \bb S^{n-1})$. 

In Sec.~\ref{sec:experiments}, we illustrate this result in the case of phase-only complex Gaussian random measurements of sparse signals, with a recovery of their direction ensured by \ref{eq:BP-def}. In particular, we show there that PO-CS requires about twice the number of measurements required for perfect signal recovery in linear CS. This is naively expected as the model \eqref{eq:PO-CS} provides $m$ constraints (associated with $m$ phases) compared to \eqref{eq:cs-model} that delivers $2m$ independent observations.  

\section{Robust and Stable Signal Direction Estimation}
\label{sec:noisy-phase-sensing}

Recovering a low-complexity signal from its phase-only observations can be made both robust and stable: we can allow for some noise in the PO-CS model \eqref{eq:PO-CS}, as well as some signal model mismatch ---~assuming $\bs x \notin \cl K$ with small $e_0(\bs x, \cl K)$ in \eqref{eq:l2-l1-inst-opt-gen} --- while still accurately estimating the direction of $\bs x$.

Given an unknown signal $\bs x \in \bb R^n$, we consider the noisy PO-CS model 
\begin{equation}
  \label{eq:PO-CS-noisy}
  \bs z = \signc (\bs A \bs x) + \bs \epsilon = \bs z_0 + \bs \epsilon,
\end{equation}
where $\bs \epsilon \in \bb C^{m}$ stands for some bounded noise with $\|\bs \epsilon\|_\infty \leq \tau$ for some level $\tau > 0$.
We note that this model encompasses a direct disturbance of the phase of \eqref{eq:PO-CS}; if $\bs z = e^{\im \bs \xi} \odot \signc (\bs A \bs x)$, where the exponential applies componentwise, $\odot$ is the Hadamard product, and $\|\bs \xi\|_\infty \leq \tau$, then $\bs z = \signc (\bs A \bs x)  + \bs \epsilon$, with $\bs \epsilon := (e^{\im \bs \xi} - 1) \odot \signc (\bs A \bs x)$ such that $\|\bs \epsilon\|_\infty \leq \tau$ (since $|e^{\im \mu} - 1| \leq |\mu|$ for $\mu \in \bb R$). 
\medskip

As the model \eqref{eq:PO-CS-noisy} does not depend on the signal amplitude, we can still work under the hypothesis that $\|\bs A \bs x\|_1 = \kappa \sqrt m$ for some $\kappa > 0$. From Lemma~\ref{lem:concentration-L1Phi-norm} and \eqref{eq:concentration-L1Phi-norm}, this shows that, for $\bs A = \bs \Phi/\sqrt m$ with $\bs \Phi \sim \cl N_{\bb C}(0,2)$, $\delta > 0$ and $\kappa = \sqrt{\pi/2}$, $(1+\delta)^{-1} \leq \|\bs x\| \leq (1-\delta)^{-1}$ with probability exceeding $1 - C \exp(-c \delta^2 m)$.

Under this hypothesis, given the definition of $\bs A_{\bs v}$ for any $\bs v \in \bb C^m$ in \eqref{eq:equiv-cs-model}, we find
$$
\bs A_{\bs z} \bs x = \bs A_{\bs z_0} \bs x + \bs A_{\bs \epsilon} \bs x = \bs e_1 + \bs \epsilon_{\bs x},
$$
with $\bs \epsilon_{\bs x} := \bs A_{\bs \epsilon} \bs x$. In other words, if $\|\bs \epsilon_{\bs x}\| \leq \varepsilon$ for some noise level $\varepsilon > 0$, the signal $\bs x$ respects the following $\ell_2$-fidelity constraint:
\begin{equation}
  \label{eq:1}
  \|\bs A_{\bs z} \bs u - \bs e_1\| \leq \varepsilon.  
\end{equation}
Moreover, in the case where $\bs A$ is a normalized complex Gaussian matrix, it is easy to estimate $\varepsilon$. 
\begin{lemma}
  Given a signal $\bs x \in \bb R^n$, and its noisy PO-CS measurements $\bs z = \signc(\bs A \bs x) + \bs \epsilon$ with $\bs A = \bs \Phi/\sqrt m$, $\bs \Phi \sim \cl N^{m \times n}_{\bb C}(0,2)$, and $\|\bs \epsilon\|_\infty \leq \tau$, if $\|\bs A \bs x\|_1 = \kappa \sqrt m$ with $\kappa = \sqrt{\pi/2}$, then, given $\delta > 0$, 
  $$
  \ts \|\bs A_{\bs \epsilon} \bs x\| \leq \sqrt 2 \tau \frac{1+\delta}{1-\delta}
  $$
  with probability exceeding $1-C\exp(-c \delta^2 m)$.
\end{lemma}
\begin{proof}
  If $\bs \Phi \sim \cl N^{m \times n}_{\bb C}(0,2)$, Lemma~\ref{lem:concentration-L1Phi-norm} and the $\ell_2$-norm concentration of Gaussian random projections~\cite{Ver12,baraniuk2008simple} show that we have jointly
$$
\ts  | \inv{\kappa m} \|\bs \Phi \bs x\|_1  - \|\bs x\| \big| \leq \delta \|\bs x\|,\quad | \inv{m} \|\bs \Phi \bs x\|^2  - \|\bs x\|^2 \big| \leq \delta \|\bs x\|^2,
$$
with probability exceeding $1-C\exp(-c \delta^2 m)$. Under this event, for $\bs A = \bs \Phi / \sqrt m$, and under the hypothesis $\sqrt m \|\bs A \bs x\|_1 = \|\bs \Phi \bs x\|_1 = \kappa m$, we find from \eqref{eq:equiv-cs-model}
\begin{align*}
  \ts \|\bs A_{\bs \epsilon} \bs x\|^2&\ts = |\scp{\bs \alpha_{\bs \epsilon}}{\bs x}|^2 + \|\bs H_{\bs \epsilon} \bs x\|^2 = \inv{\kappa^2 m} |\scp{\bs \epsilon}{\bs A\bs x}|^2 + \|\Im(\diag(\bs \epsilon)^* \bs A \bs x)\|^2\\
                               &\ts \leq \frac{\tau^2}{\kappa^2 m}\|\bs A \bs x\|^2_1 + \tau^2 \|\bs A \bs x\|^2 \leq \tau^2 (1 + (1+\delta) \|\bs x\|^2) \leq \tau^2 (1 + \frac{1+\delta}{(1-\delta)^2\kappa^2 m^2} \|\bs \Phi\bs x\|_1^2)\\
  &\ts \leq \frac{2(1+\delta)^2}{(1-\delta)^2} \tau^2.
\end{align*}
where we used Holder's inequality for the first inequality. This concludes the proof.

\end{proof}

From these observations, the definition of recoverable set in Sec.~\ref{sec:sensing-model} provides the following theorem.

\begin{theorem}[Signal direction estimation in noisy PO-CS]  
\label{thm:sig-dir-rec-noisy-po-cs}
Let $\cl K \subset \bb R^n$ be a symmetric cone that is $(\Delta_{\cl K},\delta_0)$-recoverable with the instance optimal algorithm $\Delta_{\cl K}$. Given $\bs A \in \bb C^{m \times n}$, the signal $\bs x \in \bb R^n$ observed from its $m$ noisy PO-CS measurements $\bs z = \signc(\bs A \bs x) + \bs \epsilon$, and the hypothesis that $\|\bs A \bs x\|_1 = \kappa \sqrt m$ for some $\kappa > 0$, if the matrix $\bs A_{\bs z} \in \bb R^{(m+2)\times n}$ built from $\bs A$ and $\bs z$ respects the RIP$(\cl K - \cl K, \delta)$ with $0<\delta < \delta_0$, and if $\|\bs A_{\bs \epsilon} \bs x\| \leq \varepsilon$ for some $\varepsilon >0$, then the vector $\hat{\bs x} = \Delta_{\cl K}(\bs e_1, \bs A_{\bs z}; \varepsilon)$ estimates $\bs x$ according to the bound~\eqref{eq:l2-l1-inst-opt-gen}.  
\end{theorem}

Interestingly, the RIP condition imposed on $\bs A_{\bs z}$ by this theorem can be met for complex Gaussian random matrices. Indeed, in the case where $\bs A = \bs \Phi/\sqrt m$ with $\bs \Phi \sim \cl N_{\bb C}^{m \times n}(0, 2)$, if $m$ is sufficiently large and the noise level $\tau$ is small enough, the following theorem shows that \whp the matrix $\bs A_{\bs z} = \bs A_{\bs z_0} + \bs A_{\bs \epsilon}$ respects the RIP over $\cl K$ with a constant reduced by $\tau$.

\begin{theorem}
\label{thm:noisy-rip}
  Given $\delta,\tau >0$ with $0<\delta + 9\tau <1$, $\bs x \in \bb R^n$, $\bs A = \bs \Phi/\sqrt m$ with $\bs \Phi \sim \cl N_{\bb C}(0,2)$, $\bs \epsilon \in \bb C^m$ with $\|\bs \epsilon\|_\infty \leq \tau$, and $\bs z = \signc(\bs A \bs x) + \bs \epsilon$, if
  \begin{equation}
    \label{eq:samp-complex-noisy-sensing-mtx}
    m \geq C (1+\delta^{-2}) w\big((\cl K - \bb R \bs x) \cap \bb S^{n-1}\big)^2,
  \end{equation}
  then, the matrix $\bs A_{\bs z}$ defined in \eqref{eq:Ax-def} from $\bs z$ and $\bs A$ with $\kappa = \sqrt{\pi/2}$ satisfies the RIP$(\cl K, \delta + 9\tau)$ with probability exceeding $1 - C \exp(-c \delta^2 m)$.
\end{theorem}

Therefore, in the case of PO-CS associated with a complex Gaussian random matrix, and if the symmetric cone $\cl K$ is $(\Delta_{\cl K},\delta_0)$-recoverable, as soon as $m$ is large compared to $(1+\delta^{-2}) w\big((\cl K - \bb R \bs x) \cap \bb S^{n-1}\big)^2$ with $\delta + 9 \tau \leq \delta_0$, then, Thm~\ref{thm:sig-dir-rec-noisy-po-cs} and Thm~\ref{thm:noisy-rip} tell us that we can stably and robustly recover the direction of $\bs x \in \bb R^n$ --- in the sense of~\eqref{eq:l2-l1-inst-opt-gen}~--- from its PO-CS measurements. 

\begin{remark}
\label{rem:cond-tau}
While we did not optimize it in the corresponding proof, the condition imposed on $\tau$ in Thm.~\ref{thm:noisy-rip} makes sense. Indeed, combining Thm~\ref{thm:sig-dir-rec-noisy-po-cs} and Thm~\ref{thm:noisy-rip} we conclude that the estimation of the signal direction is possible if $\tau < \delta_0/9 = O(1)$ provided $m$ is large enough (as imposed by setting $\delta = \delta_0 - 9\tau$ in \eqref{eq:samp-complex-noisy-sensing-mtx}).
Moreover, for the noisy PO-CS model $\bs z = e^{\im \bs \xi} \odot \signc(\bs A \bs x)$, if $\xi_k \sim_{\iid} \cl U([-\pi, \pi])$ for $k \in [m]$, then the phase of each component of $\bs z$ is uniformly distributed over $[-\pi, \pi]$ whatever the value of $\bs x$; estimating $\bs x/\|\bs x\|$ is then impossible. In this case, we have $\tau = {\|(e^{\im \bs \xi} - \bs 1) \odot \signc(\bs A \bs x)\|_\infty} = \|e^{\im \bs \xi} - \bs 1\|_\infty > \|\bs \xi\|_\infty$, with $\|\bs \xi\|_\infty$ arbitrary close to $\pi$ (\whp, for large value of $m$).  Although
the question of the existence of a robust algorithm for $\tau \in [\delta_0/9, \pi]$ remains open, we thus see that imposing $\tau = O(1)$ to recover the signal direction is realistic.
\end{remark}

\section{Proofs}
\label{sec:proofs}

This section is devoted to proving the RIP of a matrix built in \eqref{eq:Ax-def} from a complex Gaussian random matrix $\bs A$ and a signal $\bs x \in \bb R^n$, in the context of noiseless and noisy PO-CS. We first introduce a few useful lemmata. 

\subsection{Auxiliary lemmata}
\label{sec:auxilliary-lemmatas}

We first need this classical result from Ledoux and Talagrand
\cite[Eq. 1.6]{ledoux2013probability}. 
\begin{lemma}
\label{lem:LD-concent}
If the function $F: \bb R^n \to \bb R$ is Lipschitz with constant $\lambda$, \ie $|F(\bs u) - F(\bs u')| \leq \lambda \|\bs u - \bs u'\|$ for all $\bs u, \bs u' \in \bb R^n$, then, for $r>0$ and $\bs \gamma \sim \mathcal{N}^{n}(0,1)$, 
\begin{equation}
\ts \bb P \big( \big| F(\bs \gamma) -\bb E\big(F(\bs \gamma)\big) \big| > r\big) \leq 2 \exp(-\frac{1}{2} r^2 \lambda^{-2}).
\end{equation}
\end{lemma}

The following lemma characterizes the concentration of the random variable $\|\bs \Phi \bs z\|_1$ for a complex Gaussian random matrix $\bs \Phi$ given $\bs z \in \bb R^{n}$. 
\begin{lemma}
  \label{lem:concentration-L1Phi-norm}
 Given $\delta >0$, $\bs z \in \bb R^n$, and $\bs \Phi \sim \cl N_{\bb C}^{m \times n}(0,2)$, we have
  \begin{equation}
    \label{eq:concentration-L1Phi-norm}
\ts \bb P\big[\big| \inv{\kappa m} \|\bs \Phi \bs z\|_1  - \|\bs z\| \big| > \delta \|\bs z\| \big] \leq C \exp(-c \delta^2 m).  
  \end{equation} 
\end{lemma}
\begin{proof}
  By homogeneity of \eqref{eq:concentration-L1Phi-norm}, we can assume that $\|\bs z\|=1$. In this case, $\bs \Phi \bs z \sim \bs g = \bs g^\Re + \im \bs g^\Im$, with  $\bs g^\Re, \bs g^\Im \sim_{\iid} \cl N_{\bb C}^{m}(0, 1)$. In this context, $\inv{\kappa m} \|\bs \Phi \bs z\|_1 \sim \inv{\kappa m}\|\bs g^\Re + \im \bs g^\Im\|_1 = \sum_k X_k$, where each independent random variable $X_k := \big ( (g_k^\Re)^2 + (g_k^\Im)^2\big)^{1/2}$ follows a Rayleigh distribution with unit scale parameter~\cite{papoulis2002probability}, which gives $\bb E X_k = \sqrt{\pi/2} = \kappa$. Therefore, since each $X_k$ is a sub-Gaussian random variable~\cite{Ver12}, the concentration property of the sum of $m$ such random variables provides~\cite{Ver12,FDVJ19}, for $t>0$, 
$$
\ts \bb P\big[\,\big|\,\|\bs g^\Re + \im \bs g^\Im\|_1 - \kappa m\big| = \big|\sum_k X_k - \bb E X_k\big| > t \kappa\,\big]\ \leq\ C \exp(-c \frac{t^2}{m}).
$$
The result follows from a simple change of variable.  
\end{proof}

We now show that the Gaussian mean width of a bounded set projected on a subspace is bounded by twice the width of the original set.
\begin{lemma}
\label{lem:gmw-project-space}
For a bounded subset $\cl K' \subset \bb R^n$, and a $L$-dimensional subspace $\cl L \subset \bb R^n$ related to the projecting matrix $\bs P \in \bb R^{L \times n}$, we have
\begin{equation}
  \label{eq:gmw-proj-bound}
w(\bs P (\cl K')) \leq 2 w(\cl K').  
\end{equation}
\end{lemma}

\begin{proof}
  Let us define $\cl L^\perp$ and $\bs P^\perp$ the subspace orthogonal to $\cl L$ and its $(n-L) \times n$ projection matrix, respectively, so that $(\bs P^\perp)^\top \bs P^\perp + \bs P^\top \bs P = \Id_n$. Then, for $\bs g \sim \cl N^{L}(0,1)$, which can always be written as $\bs g = \bs P \bs g'$ for some $\bs g' \sim \cl N^{n}(0,1)$ (by marginalization of the Gaussian distribution), we have 
  \begin{align*}
    w(\bs P \cl K')&\ts = \bb E \sup_{\bs u \in \cl K'} \scp{\bs g}{\bs P \bs u} = \bb E \sup_{\bs u \in \cl K'} \scp{\bs P\bs g'}{\bs P \bs u} = \bb E \sup_{\bs u \in \cl K'} \scp{\bs P^\top \bs P\bs g'}{\bs u}\\
                    &\ts \leq \bb E \sup_{\bs u \in \cl K'} \max_{\sigma \in \{\pm 1\}} \sigma \scp{(\bs P^\perp)^\top \bs P^\perp \bs g'}{\bs u} + \scp{\bs P^\top \bs P\bs g'}{\bs u}\\
                    &\ts \leq \bb E \sup_{\bs u \in \cl K'} \sum_{\sigma \in \{\pm 1\}} \scp{\sigma   (\bs P^\perp)^\top \bs P^\perp \bs g' + \bs P^\top \bs P\bs g'}{\bs u}\\
                    &\ts \leq \sum_{\sigma \in \{\pm 1\}} \bb E \sup_{\bs u \in \cl K'} \scp{\sigma   (\bs P^\perp)^\top \bs P^\perp \bs g' + \bs P^\top \bs P\bs g'}{\bs u}\ = 2 w(\cl K'),
  \end{align*}
  where the last line used the fact that $\pm (\bs P^\perp)^\top \bs P^\perp \bs g' + \bs P^\top \bs P\bs g' \sim \cl N^n(0,1)$.
\end{proof}

Finally, this fourth lemma introduces a complex, local sign-product embedding~\cite{plan2012robust,jacques2013quantized}; it shows that projecting the random measurements of any low-complexity vector onto the phase-only measurements of another fixed vector is close, up to a rescaling, to the scalar product between both vectors.
\begin{lemma}
  \label{lem:sign-product-embedding}
  Given a symmetric cone $\cl K \subset \bb R^n$, $\delta > 0$, and $\bs z \in \bb S^{n-1}$, if
  \begin{equation}
    \label{eq:sign-product-embedding}
\ts  m \geq C \delta^{-2} w^2(\cl K \cap \bb S^{n-1}),   
  \end{equation}
then, with probability exceeding $1 - C \exp(-c \delta^2 m)$, the random matrix $\bs \Phi \sim \cl N_{\bb C}^{m \times n}(0,2)$ satisfies
  \begin{equation}
    \label{eq:srp-def}
    \ts |\inv{\kappa m} \scp{\signc(\bs \Phi \bs z)}{\bs \Phi \bs u}  - \scp{\bs z}{\bs u}| \leq \delta \|\bs u\|, \quad \forall \bs u \in \cl K. 
  \end{equation}
\end{lemma}

\begin{proof}
Since \eqref{eq:srp-def} is homogeneous in $\bs u$ and $\cl K$ is conic, we can assume $\|\bs u\|=1$. Given $\cl K^* = \cl K \cap \bb S^{n-1}$, we must prove that $\bb P(E > \delta) \leq C \exp(-c \delta^2 m)$ with 
$$
\ts E := \sup_{\bs u \in \cl K^*} |\inv{\kappa m} \scp{\signc(\bs \Phi \bs z)}{\bs \Phi \bs u} - \scp{\bs z}{\bs u}|.
$$
Using the decomposition $\bs u = \bs u^{\parallel} + \bs u^\perp$ with $\bs u^{\parallel} := \scp{\bs z}{\bs u} \bs z$ and $\bs u^\perp := (\bs u - \scp{\bs z}{\bs u} \bs z)$ (with $\scp{\bs u^{\parallel}}{\bs u^\perp}=0$), and the triangular inequality, we find 
\begin{align*}
E\ &\ \ts \leq \sup_{\bs u \in \cl K^*} |\scp{\bs z}{\bs u}|\,\big |\inv{\kappa m}\|\bs \Phi \bs z\|_1 - 1\big|\, +\, \inv{\kappa m} \sup_{\bs u \in \cl K^*} |\scp{\signc(\bs \Phi \bs z)}{\bs \Phi \bs u^\perp}|\\  
                     &\ \leq \underbrace{\ts \big|\inv{\kappa m}\|\bs \Phi \bs z\|_1 - 1\big|}_{A(\bs \Phi)} + \underbrace{\ts \inv{\kappa m} \sup_{\bs u \in \cl K^*} |\scp{\signc(\bs \Phi \bs z)}{\bs \Phi \bs u^\perp}|}_{B(\bs \Phi)}.
\end{align*}

Moreover, since $\bb P( E > \delta) \leq \bb P( 2\max(A,B) > \delta) \leq \bb P(A > \delta/2) + \bb P(B > \delta/2)$,
and $\bb P(A > \delta/2) \leq C \exp(-c \delta^2 m)$ from Lemma~\ref{lem:concentration-L1Phi-norm}, we only have to prove that $\bb P(B > \delta/2) \leq C \exp(-c \delta^2 m)$ if \eqref{eq:sign-product-embedding} holds.

From the rotational invariance of the Gaussian distribution, $\bs \Phi$ and $\bs \Phi \bs R$ have the same distribution for any rotation matrix $\bs R \in \bb R^{n \times n}$, and $\bb P(B(\bs \Phi) > \delta/2) = \bb P(B(\bs \Phi\bs R) > \delta/2)$. Since $\bs z \in \bb S^{n-1}$ is fixed, we decide to set $\bs R$ such that $\bs R \bs z = \bs e_1$ and $(\bs R \bs v)_1 = 0$ for all $\bs v \in (\bb R\bs z)^\perp := \{\bs u \in \bb R^n: \scp{\bs u}{\bs z} = 0\}$. Then, defining the restriction matrix $\bs S \in \bb R^{(n-1) \times n}$ such that $\bs S \bs w = (w_2,\, \cdots, w_n)^\top \in \bb R^{n-1}$ for any $\bs w \in \bb R^n$, we observe that, for all $\bs v \in (\bb R\bs z)^\perp$, $\bs \Phi \bs R \bs v = \bs G \bs S \bs R \bs v$ with $\bs G = \bs \Phi \bs S^\top \sim \cl N_{\bb C}^{m\times (n-1)}(0,2)$ independent of $\bs g := \bs \Phi \bs R \bs z = \bs \Phi \bs e_1 \sim \cl N_{\bb C}^{m}(0,2)$. 

Therefore, since $\bs u^\perp \in (\bb R\bs z)^\perp$ and $\bs S \bs R \bs u^\perp = \bs S \bs R \bs u$ by design of $\bs S$ and $\bs R$, we can write 
\begin{equation}
\label{eq:equiv-spe-dist}
  \ts |\scp{\signc(\bs \Phi \bs R \bs z)}{\bs \Phi \bs R \bs u^\perp}| = |\scp{\signc(\bs g)}{\bs G \bs S \bs R \bs u^\perp}| = |\scp{\signc(\bs g)}{\bs G \bs S \bs R \bs u}|,
\end{equation}
From the independence of $\bs g$ and $\bs G$, the equivalence \eqref{eq:equiv-spe-dist} allows us to condition the random variable $B(\bs \Phi \bs R)$ to the value of~$\bs g$ while preserving the distribution of $\bs G$. We thus focus on bounding  $P\big(B(\bs \Phi \bs R) > {\delta}/{2}\ |\,\bs g\big)$ and eventually use $\bb P\big(B(\bs \Phi \bs R) > {\delta}/{2}\big) = \bb E \bb P\big(B(\bs \Phi \bs R) > {\delta}/{2}\ |\,\bs g\big)$ by expectation over $\bs g$.

For $\bs g$ fixed, $|\scp{\signc(\bs g)}{\bs G \bs S \bs R \bs u}| = |\scp{\bs G^*\signc(\bs g)}{\bs S \bs R \bs u}|$ is distributed as $\sqrt m\,|\scp{\bs \gamma}{\bs S \bs R \bs u}|$ with $\bs \gamma \sim \cl N^{n-1}_{\bb C}(0,2)$, since $\|\signc(\bs g)\| = \sqrt m$. Therefore, using $|\scp{\bs a}{\bs b}| \leq 2 \max(|\scp{\bs a^\Re}{\bs b}|, |\scp{\bs a^\Im}{\bs b}|$ for any $\bs a \in \bb C^d$ and $\bs b \in \bb R^d$, we find 
\begin{align*}
  \ts \bb P\big(B(\bs \Phi \bs R) > \frac{\delta}{2}\ |\,\bs g\big)&\ts = \bb P\big( \frac{1}{\kappa \sqrt m} \sup_{\bs u \in \cl K^*} |\scp{\bs \gamma}{\bs S \bs R \bs u}|\, \geq\, \frac{\delta}{2}\ |\,\bs g\big)\\
  &\ts \leq \bb P\Big(\frac{1}{\kappa \sqrt m} \sup_{\bs u \in \cl K^*} \max\big(\,|\scp{\bs \gamma^\Re}{\bs S \bs R \bs u}|,\, |\scp{\bs \gamma^\Im}{\bs S \bs R \bs u}|\big) \, \geq\,  \frac{\delta}{4}\ |\,\bs g\Big)\\
  &\ts \leq 2 \bb P\big( F(\bs g')\, \geq\,  \frac{\delta}{4}\ |\,\bs g\big),
\end{align*}
where we defined $F: \bs v \in \bb R^{n-1} \to F(\bs v) := \sup_{\bs u \in \cl K^*} \,|\scp{\bs v}{\bs S \bs R \bs u}|/(\kappa \sqrt m)$, and we used the union bound and the fact that $\bs \gamma^\Re, \bs \gamma^\Im \sim_\iid \bs g' \sim \cl N^{n-1}(0,1)$.

Let us characterize the random variable $F(\bs g')$.  We first observe that
\begin{equation*}
  \ts \bb E F(\bs g') = \frac{1}{\kappa \sqrt m} \bb E \sup_{\bs u \in \cl K^*} \,|\scp{\bs g'}{\bs S \bs R \bs u}| = \frac{1}{\kappa \sqrt m} \bb E \sup_{\bs u \in \cl K^*} \,\scp{\bs g'}{\bs S \bs R \bs u} = \frac{1}{\kappa \sqrt m} w(\cl K'),  
\end{equation*}
where $\cl K' := (\bs S \bs R\,\cl K^*)$, the width $w$ is defined in \eqref{eq:GMW-def}, and we used the symmetry of $\cl K$.

Since $\bs S$ is a projector with $\bs S^\top \bs S + \bs e_1 \bs e_1^\top = \Id_n$, Lemma~\ref{lem:gmw-project-space} tells us that
\begin{equation}
  \label{eq:tmp1}
  \ts \kappa \sqrt m\, \bb E F(\bs g') \leq w(\cl K') \leq 2 w(\bs R \cl K^*) = 2 w(\cl K^*),  
\end{equation}
since the Gaussian mean width is invariant under rotation.

Moreover, $F$ is Lipschitz with constant $\lambda = 1/(\kappa \sqrt m)$ since $\|\bs S \bs R \bs u\| \leq 1$ for all $\bs u \in \cl K^*$,
$$
\ts |F(\bs v) - F(\bs v')| \leq \frac{1}{\kappa \sqrt m} \sup_{\bs u \in \cl K^*} \,|\scp{\bs v - \bs v'}{\bs S \bs R \bs u}| \leq \frac{1}{\kappa \sqrt m} \|\bs v - \bs v'\|.
$$

We can thus invoke Lemma~\ref{lem:LD-concent} on $F$ and $\bs g'$ to conclude that, from \eqref{eq:tmp1} and for $r > 0$,
\begin{align*}
  &\ts \bb P \big( \frac{1}{\kappa \sqrt m} \bb E \sup_{\bs u \in \cl K^*} \,|\scp{\bs g'}{\bs S \bs R \bs u}|  \geq \frac{2}{\kappa \sqrt m} w(\cl K^*) + r\ |\,\bs g\big)\\
  &\ts \leq \bb P \big( \frac{1}{\kappa \sqrt m} \bb E \sup_{\bs u \in \cl K} \,|\scp{\bs g'}{\bs S \bs R \bs u}|  \geq \bb E F(\bs g')+ r\ |\,\bs g\big) = \bb P \big( F(\bs g')  \geq \bb E F(\bs g')+ r\ |\,\bs g\big)\\
  &\ts \leq \bb P \big( \big| F(\bs g') -\bb E F(\bs g') \big| > r\ |\,\bs g\big) \leq 2 \exp(-\frac{1}{2} r^2 \kappa^2 m).  
\end{align*}
Taking, \eg $r=\delta /8$ and $m \geq 64 \,\delta^{-2} \kappa^{-2} w(\cl K^*)$, gives $r + \frac{2}{\kappa \sqrt m} w(\cl K^*) \leq \delta /4$ so that, by expectation over $\bs g$, the proof is concluded from
$$
\ts \bb P(B > \frac{\delta}{2} ) = \bb E\,\bb P(B > \frac{\delta}{2}\,|\,\bs g) \leq 2 \bb E\,\bb P \big( \frac{1}{\kappa \sqrt m} \bb E \sup_{\bs u \in \cl K^*} \,|\scp{\bs g'}{\bs S \bs R \bs u}|  \geq \frac{\delta}{4}\ |\bs g\big) \leq C \exp(-c \delta^2 m).
$$
\end{proof}

\subsection{Proof of Thm~\ref{thm:rip-for-Ax}}
\label{sec:proof-rip-Ax}

Given a symmetric cone $\cl K \in \bb R^n$, $\bs x \in \bb R^n$ (with $\bar{\bs x} := \bs x / \|\bs x\|$), $\bs A = \bs \Phi/\sqrt m$ with $\bs \Phi \sim \cl N_{\bb C}^{m \times n}(0,2)$, and $\bs z = \signc(\bs A \bs x)$, we can now determine under which conditions and with which probability the matrix $\bs A_{\bs z}$ defined in \eqref{eq:equiv-cs-model} satisfies the RIP over $\cl K$ for some distortion $\delta>0$. This amounts to showing that, under the condition \eqref{eq:sample-complexity} and with probability exceeding $1 - C \exp(-c\delta^2 m)$,
$$
\ts \big| \|\bs A_{\bs z} \bs u\|^2 - 1 \big| \leq \delta,\quad \forall \bs u \in \cl K^* := \cl K \cap \bb S^{n-1},
$$
where we can assume $\|\bs u\| = 1$ since $\cl K$ is a cone.
\medskip

Given the definition of $\bs \alpha_{\bs z}$ and $\bs H_{\bs z}$ in \eqref{eq:equiv-cs-model}, for $\bs u \in \cl K^*$, we see that
\begin{align}
  \ts \|\bs A_{\bs z} \bs u\|^2&\ts =  \big\| \big[ \bs \alpha_{\bs z}^\Re, \bs \alpha_{\bs z}^\Im \big]^\top \bs u \big\|^2 + \|\bs H_{\bs z} \bs u\|^2 = |\scp{\bs \alpha_{\bs z}}{\bs u}|^2 + \|\bs H_{\bs z} \bs u\|^2\nonumber\\
  \label{eq:devel-norm-Axu}
  &\ts = \inv{\kappa^2 m^2} |\scp{\signc(\bs \Phi \bs x)}{\bs\Phi\bs u}|^2 + \|\bs H_{\bs z} \bs u\|^2.
\end{align}

We first bound the first term of \eqref{eq:devel-norm-Axu}. Under the condition $m \geq C (1+\delta^{-2}) w^2(\cl K \cap \bb S^{n-1})$, which is involved by \eqref{eq:sample-complexity}, Lemma~\ref{lem:sign-product-embedding} and Thm~\ref{thm:gen-rip} inform us that the joint event 
\begin{align*}
  \ts \cl E:\quad  \ |\inv{\kappa m} \scp{\signc(\bs \Phi \bs x)}{\bs\Phi\bs v} - \scp{\bar{\bs x}}{\bs v}| \leq \frac{3}{8}\delta\quad \text{and}\quad  |\inv{m} \|\bs \Phi \bs u\|^2 - 1| \leq \kappa^2 -1,\ \forall \bs v \in \cl K^*,
\end{align*}
holds (by union bound) with probability exceeding $1 - C \exp(-c \delta^2 m)$. Since for any $A, B \in \bb C$, $| |A|^2 - |B|^2| \leq |A +B||A-B|$, $\cl E$ involves
\begin{equation}
  \label{eq:first-ingredient}
  \ts |\inv{\kappa^2 m^2} |\scp{\signc(\bs \Phi \bs x)}{\bs\Phi\bs u}|^2 - |\scp{\bar{\bs x}}{\bs u}|^2|
  \leq \frac{3\delta}{8} (1 + |\inv{\kappa m} \scp{\signc(\bs \Phi \bs x)}{\bs\Phi\bs u}|) \leq \frac{3\delta}{8} (1 + \frac{1}{\kappa \sqrt m} \|\bs \Phi \bs u\|)  \leq \frac{3}{4} \delta.  
\end{equation}

We now focus on bounding the second term of \eqref{eq:devel-norm-Axu}.  We first note that for any $\bs u \in \cl K^*$ decomposed as $\bs u = \bs u^\parallel + \bs u^\perp$ with $\bs u^\parallel = \scp{\bs u}{\bar{\bs x}}  \bar{\bs x}$ and $\bs u^\perp = (\bs u - \scp{\bs u}{\bar{\bs x}}  \bar{\bs x})$, we have $\bs H_{\bs z} \bs u = \bs H_{\bs z} \bs u^\perp$ since $\bs H_{\bs z} \bs x = \bs 0$. We are going to show that
\begin{equation}
  \label{eq:bound-on-px}
  \ts p_{\bs z}(\bs \Phi) := \bb P\big[\exists \bs u \in \cl K^*, \big| \|\bs H_{\bs z} \bs u^\perp\|^2 - \|\bs u^\perp\|^2 \big| > \inv{4}\delta \|\bs u^\perp\|^2 \big] \leq C \exp(-c \delta^2 m),
\end{equation}
provided \eqref{eq:sample-complexity} is satisfied.

Using the rotational invariance of the Gaussian distribution, we note that $p_{\bs z}(\bs \Phi) = p_{\bs z}(\bs \Phi \bs R)$ for any rotation matrix $\bs R \in \bb R^{n \times n}$. We proceed similarly to the proof of Lemma~\ref{lem:sign-product-embedding} and take $\bs R$ such that $\bs R \bs x = \|\bs x\| \bs e_1$ and $(\bs R \bs v)_1 = 0$ for all $\bs v \in (\bb R\bs x)^\perp$; we thus find that for all $\bs v \in (\bb R\bs x)^\perp$, $\bs \Phi \bs R \bs v = \bs G \bs S\bs R \bs v$ with $\bs G = \bs \Phi \bs S^\top \sim \cl N_{\bb C}^{m\times (n-1)}(0,2)$ independent of $\bs g := \|\bs x\|^{-1} \bs \Phi \bs R \bs x = \bs \Phi \bs e_1 \sim \cl N_{\bb C}^m(0, 2)$.

From the independence of $\bs g$ and $\bs G$, we can condition $p_{\bs z}$ to the value of $\bs g$ without altering the distribution of $\bs G$, and eventually computing this probability by expectation over $\bs g$ from $p_{\bs z}(\bs \Phi \bs R) = \bb E  p_{\bs z}(\bs \Phi \bs R|\bs g)$. In this context, defining $\bs z' := \signc(\bs A\bs R \bs x)$ and using the properties of $\bs R$, we have
$\|\bs u^\perp\|^2 = \|\bs R \bs u^\perp\|^2 = \|\bs S \bs R \bs u\|^2$. Moreover $\bs H_{\bs z'} \bs R \bs u^\perp =  \frac{1}{\sqrt m}\bs G' \bs S \bs R \bs u^\perp = \frac{1}{\sqrt m}\bs G' \bs S \bs R \bs u$ with $\bs G' := (\bs D_{\bs z'}^\Re \bs G^\Im  - \bs D_{\bs z'}^\Im \bs G^\Re) \sim \cl N^{m \times (n-1)}(0, 1)$ since $(\bs D_{\bs z'}^\Re)^2 + (\bs D_{\bs z'}^\Im)^2 = \Id_m$. Therefore
\begin{align*}
  \ts p_{\bs z}(\bs \Phi \bs R\,|\bs g)&\ts = \bb P\big[\exists \bs u \in \cl K^*, \big| \|\bs H_{\bs z'} \bs R \bs u^\perp\|^2 - \|\bs u^\perp\|^2 \big| > \inv{4}\delta \|\bs u^\perp\|^2 \,|\bs g\big]\\
                               &\ts = \bb P\big[\exists \bs u \in \cl K^*, \big| \frac{1}{m}\|\bs G' \bs S \bs R \bs u\|^2 - \|\bs S \bs R \bs u\|^2 \big| > \inv{4}\delta \|\bs S \bs R \bs u\|^2 \,|\bs g\big].
\end{align*}
By homogeneity, since $\cl K$ is a cone, we thus get 
\begin{align*}
  \ts p_{\bs z}(\bs \Phi \bs R)&\ts = \bb P\big[\exists \bs v \in \bs S \bs R(\cl K^*), \big| \frac{1}{m}\|\bs G' \bs v\|^2 - \|\bs v\|^2 \big| > \inv{4}\delta \|\bs v\|^2 \,|\bs g\big]\\
      &\ts = \bb P\big[\exists \bs v \in \bs S \bs R(\cl K), \big| \frac{1}{m}\|\bs G' \bs v\|^2 - \|\bs v\|^2 \big| > \inv{4}\delta \|\bs v\|^2 \,|\bs g\big],
\end{align*}
hence showing that $p_{\bs z}(\bs \Phi \bs R)$ is the failing probability of $\frac{1}{\sqrt m} \bs G'$ satisfying the RIP over $\cl K' = \bs S \bs R(\cl K)$ with constant $\delta/4$. We can thus invoke Thm.~\ref{thm:gen-rip} and see that \eqref{eq:bound-on-px} holds, \ie $p_{\bs z}(\bs \Phi \bs R) = \bb E p_{\bs z}(\bs \Phi \bs R\,|\bs g) \leq C \exp(-c \delta^2 m)$, provided that
$$
m \geq C \delta^{-2} w^2(\bs S \bs R(\cl K) \cap \bb S^{n-2}).
$$
For $\cl K^\perp := \{\bs u - \scp{\bar{\bs x}}{\bs u}\bar{\bs x}: \bs u \in \cl K\} \subset \cl K - \bb R \bs x$, the isotropy and monotonicity of the Gaussian mean width~\cite[Sec. 3.2]{chandrasekaran2012convex} involve   
$$
w(\bs S \bs R(\cl K) \cap \bb S^{n-2}) = w\big( (0 \oplus \bs S \bs R(\cl K)) \cap \bb S^{n-1}\big) = w( \bs R(\cl K^\perp) \cap \bb S^{n-1}) \leq w\big( (\cl K - \bb R \bs x) \cap \bb S^{n-1}\big), 
$$
where we define $0 \oplus \cl S := \{(0,v): v \in \cl S\} \subset \bb R^{n}$ for any set $\cl S \subset \bb R^{n-1}$.

We conclude that, provided $m \geq C \delta^{-2} w( (\cl K - \bb R \bs x) \cap \bb S^{n-1})^2$, which is verified from \eqref{eq:sample-complexity},  the event
\begin{equation}
  \label{eq:Eperp-def}
\ts \cl E_{\perp}:\quad \forall \bs u \in \cl K \cap \bb S^{n-1},\ \big|\inv{m} \|\bs H_{\bs z} \bs u^\perp\|^2 - \|\bs u^\perp\|^2 \big| \leq \inv{4}\delta \|\bs u^\perp\|^2  
\end{equation}
holds with probability exceeding $1 - C \exp(-c \delta^2 m)$.
\medskip

Finally, from \eqref{eq:first-ingredient} and \eqref{eq:Eperp-def}, if the joint event $\cl E \cap \cl E_{\perp}$ holds, which occurs with probability exceeding $1 - C \exp(-c \delta^2 m)$, then, for all $\bs u \in \cl K \cap \bb S^{n-1}$,  
\begin{align*}
  \ts \|\bs A_{\bs z} \bs u\|^2&\ts = \inv{\kappa^2 m^2} |\scp{\signc(\bs \Phi \bs x)}{\bs\Phi\bs u}|^2 + \inv{m}\|\bs H_{\bs z} \bs u\|^2\\
                               &\ts \leq |\scp{\bar{\bs x}}{\bs u}|^2 + \frac{3}{4}\delta + \|\bs u^\perp\|^2 + \inv{4}\delta\\
  &\ts = \|\bs u\|^2 + \delta = 1 + \delta.
\end{align*}
We find similarly that $\|\bs A_{\bs z} \bs u\|^2 \geq 1 - \delta$, which provides the final result. 

\subsection{Proof of Thm~\ref{thm:noisy-rip}}

From Thm~\ref{thm:rip-for-Ax}, we already know that, under the condition \eqref{eq:samp-complex-noisy-sensing-mtx}, $\bs A_{\bs z_0}$ with $\bs z_0 = \signc(\bs A \bs x)$ respects the RIP$(\cl K, \delta)$ with probability exceeding $1 - C \exp(-c \delta^2 m)$, which means that the event 
$$
\ts \cl E_0:\quad 1 - \delta \leq \|\bs A_{\bs z_0} \bs u\|^2 \leq 1 + \delta, \ \forall \bs u \in \cl K, 
$$
holds with that probability.

Moreover, from Thm~\ref{thm:gen-rip}, provided $m \geq C \delta^2 w(\cl K \cap \bb S^{n-1})$, which holds if \eqref{eq:samp-complex-noisy-sensing-mtx} is verified, and with a failing probability smaller than $C \exp(-c \delta^2 m)$, this other event is respected:
$$
\ts \cl E_1:\quad 1 - \delta \leq \frac{1}{m} \|\bs \Phi \bs u\|^2 \leq 1 + \delta, \ \forall \bs u \in \cl K, 
$$

Let us assume that both events hold, which under the condition \eqref{eq:samp-complex-noisy-sensing-mtx} happens with probability greater than $1 - C \exp(-c \delta^2 m)$. Then, from the definition of $\bs A_{\bs \epsilon}$, $\bs \alpha_{\bs \epsilon}$, and $\bs H_{\bs \epsilon}$, we compute that, for any $\bs u \in \cl K \cap \bb S^{n-1}$,
\begin{align*}
  \|\bs A_{\bs \epsilon} \bs u\|&\ts \leq \frac{1}{\kappa m} |\scp{\bs \epsilon}{\bs \Phi \bs u}|  + \|\bs H_{\bs \epsilon} \bs u\| = \frac{1}{\kappa m} |\scp{\bs \epsilon}{\bs \Phi \bs u}| + \frac{1}{\sqrt m} \|\Im(\bs D_{\bs \epsilon}^*\bs\Phi \bs u)\|\\
&\ts \leq \frac{1}{\kappa} (\frac{1}{\sqrt m}\|\bs \epsilon\|) (\frac{1}{\sqrt m}\|\bs \Phi \bs u\|)  + \frac{1}{\sqrt m} \|\bs D_{\bs \epsilon}^*\bs\Phi \bs u\|\\
                         &\ts \leq \frac{\tau}{\kappa\sqrt m}\,\|\bs \Phi \bs u\|  + \frac{\tau}{\sqrt m} \|\bs \Phi \bs u\|  \leq \tau \sqrt 2 (\inv{\kappa} + 1) = 2 \sqrt 2 \tau.
\end{align*}

Therefore, since $0 < \delta < 1$ and $\tau < 1/9$,
\begin{align*}
  \ts \|\bs A_{\bs z} \bs u\|^2&\ts \leq (\|\bs A_{\bs z_0} \bs u\| + \|\bs A_{\bs \epsilon} \bs u\|)^2 \leq (\sqrt{1 + \delta} + 2 \sqrt 2 \tau)^2\\
  &\ts = 1 + \delta + 8 \tau^2 + 4 \sqrt 2 \tau \sqrt{1+\delta} \leq 1 + \delta + 8\tau(1+ \tau) < 1 + \delta + 9 \tau,  
\end{align*}
and, since $\sqrt{1 - \delta} - 2 \sqrt 2 \tau > 0$ if $0<\delta + 9 \tau <1$, 
\begin{align*}
  \ts \|\bs A_{\bs z} \bs u\|^2&\ts \geq (\|\bs A_{\bs z_0} \bs u\| - \|\bs A_{\bs \epsilon} \bs u\|)^2 \geq (\sqrt{1 - \delta} - 2 \sqrt 2 \tau)^2\\
  &\ts = 1 - \delta - 4 \sqrt 2 \tau \sqrt{1-\delta} \geq 1 - \delta - 4\sqrt 2\tau > 1 - (\delta + 9 \tau),  
\end{align*}
which finally establishes that $\bs A_{\bs z}$ is RIP$(\cl K, \delta + 9\tau)$.

\section{Numerical Experiments}
\label{sec:experiments}

\begin{figure}[t]
    \centering
    \includegraphics[width=0.49\textwidth]{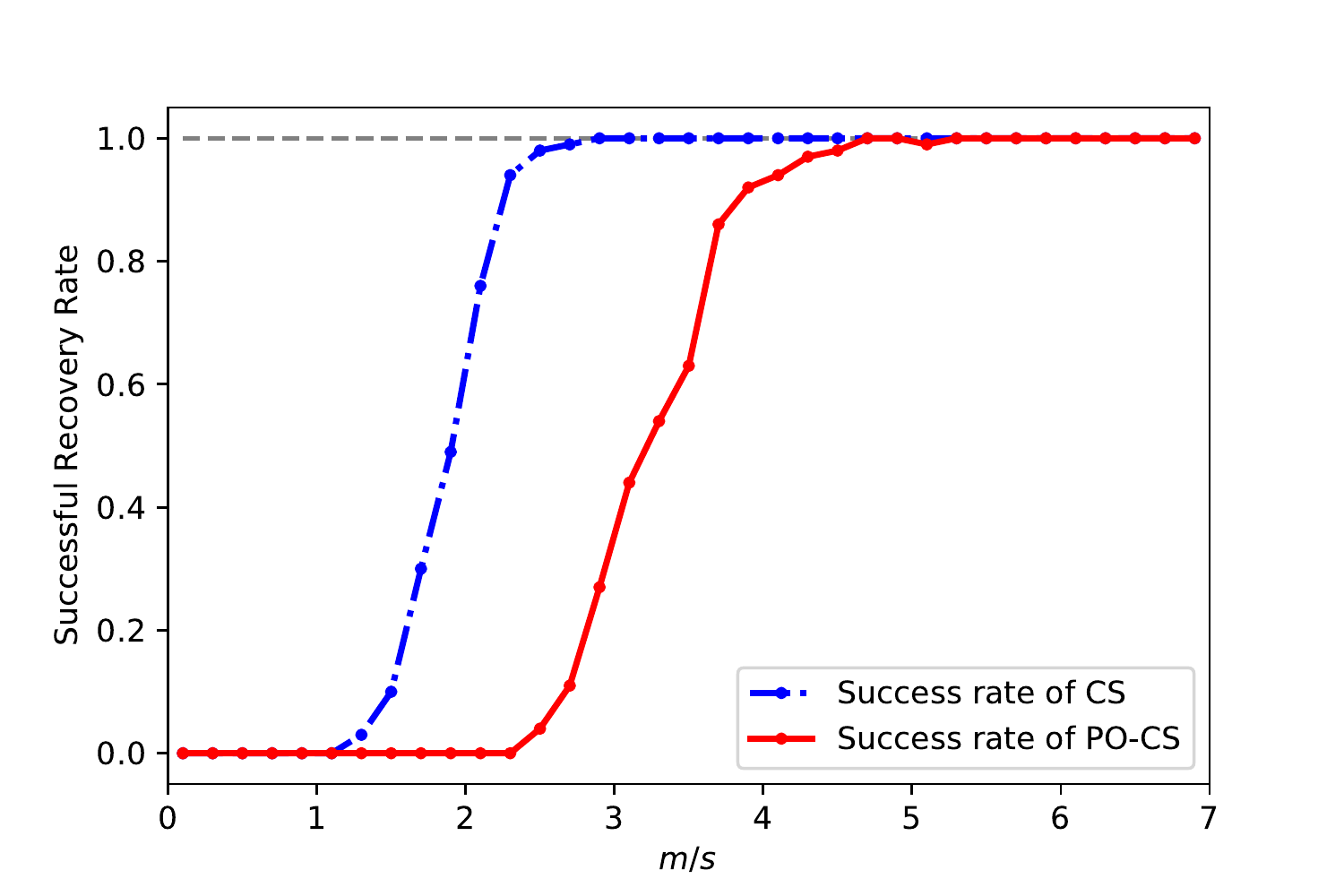}\hspace{-0.5cm}
    \includegraphics[width=0.49\textwidth]{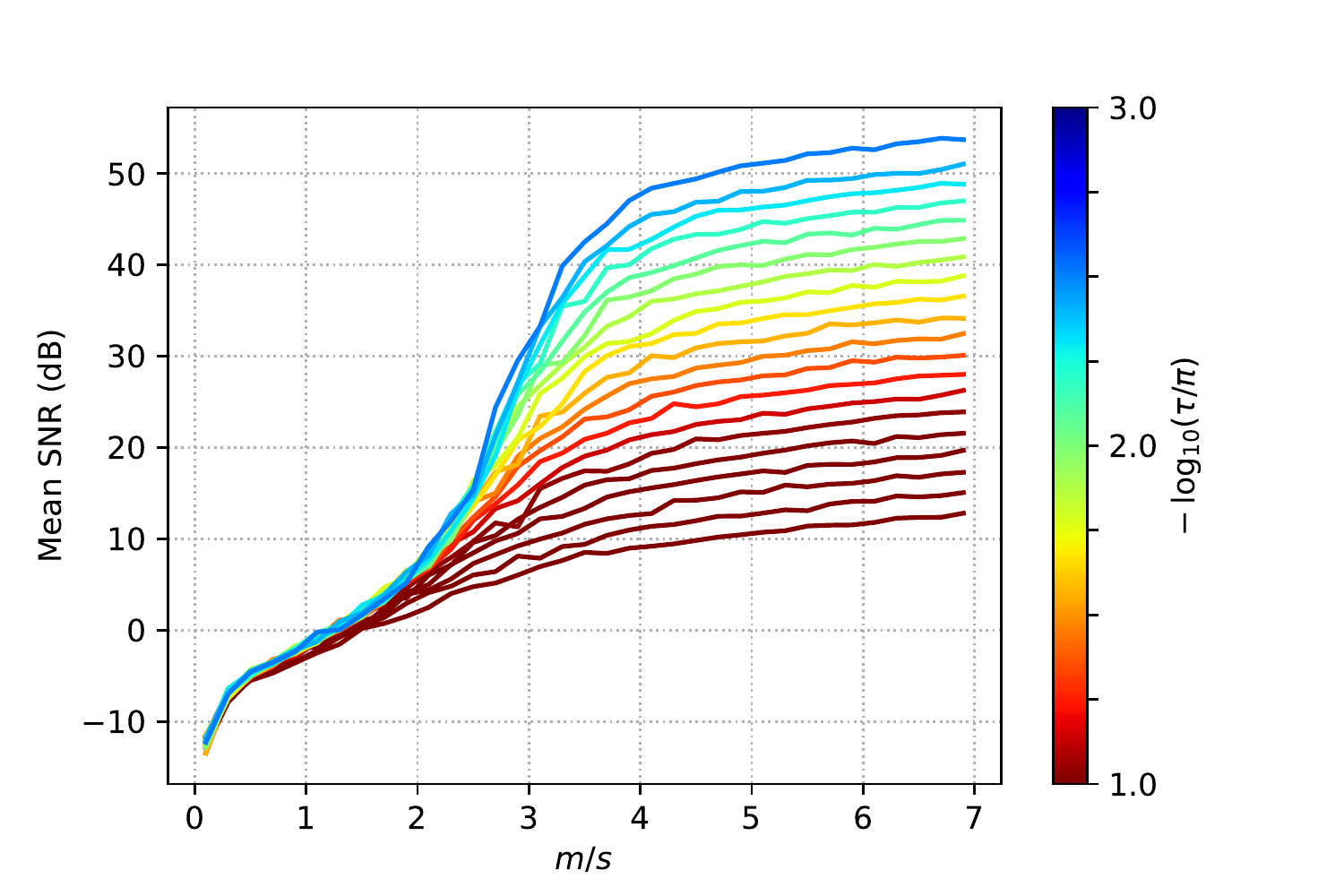}
    \caption{(left) Rate of successful recovery of $s$-sparse signals in function of $m/s$ for PO-CS (in red) under the hypothesis \eqref{eq:norm-hyp}, and in linear CS (dashed blue). (right, best viewed in color) Average SNR in function of $m/s$ for the estimation of the direction of $s$-sparse signals in PO-CS. The color coding of the SNR curves corresponds to the noise level $\tau$, as measured by $-\log_{10} \tau/\pi \in [1, 3]$. All figures were generated by Matplotlib~\cite{Hunter:2007}.}
    \label{fig:BP_rec_rate}
\end{figure}

In this section, we test the ability to recover the direction of a sparse vector from the phase of its complex compressive observations.

As a first experiment, we compare in a noiseless scenario the sample complexities of both PO-CS and (linear) CS
for sparse signal recovery. We thus study from which number of measurements a sparse vector (or its direction for PO-CS) can be perfectly recovered.  The tested experimental conditions are as follows. We have randomly generated $(s=10)$-sparse vectors in $\bb R^{(n=100)}$ for a range of measurements $m \in [1, 70]$. Each observed sparse vector $\bs x_0$ was simply generated by randomly picking a $10$-sparse support in the $100 \choose 10$ available supports, and by setting the non-zero vector components \iid as $\cl N(0,1)$. For each value of $m$, we generated the complex Gaussian random matrix as $\bs \Phi \sim \cl N_{\bb C}^{m \times n}(0,2)$. As described in Sec.~\ref{sec:sensing-model}, the $\bs x_0$ was normalized so that $\|\bs \Phi \bs x_0\|_1 = \kappa m$ with $\kappa = \sqrt{\pi/2}$. Moreover, for each $m$, the signal $\bs x_0$ and the matrix $\bs \Phi$ were randomly regenerated 100 times for reaching valid recovery statistics. 

For the reconstruction procedure, we used the basis pursuit denoising \eqref{eq:BP-def} program to solve $\Delta_{\cl K}(\bs A_{\bs z}, \bs e_1; 0)$, in comparison with the estimate produced by $\Delta_{\cl K}(\bs A, \bs y; 0)$ for the CS model $\bs y = \bs \Phi \bs x_0$. Numerically, \ref{eq:BP-def} was solved by using the Python implementation of the SPGL1 toolbox\footnote{\url{https://github.com/drrelyea/spgl1}}~\cite{van_den_Berg_2009}, together with the Numpy module~\cite{van2011numpy}. In PO-CS and in CS, a reconstruction was considered as successful when the estimate $\hat{\bs x}$ reached a relative error $\|\bs x_0 - \hat{\bs x}\|/\|\bs x_0\| \leq 10^{-3}$ (\ie a 60\,dB SNR). 

Fig.~\ref{fig:BP_rec_rate}(left) displays the success rate of both approaches, as measured from the fraction of successful reconstructions over the 100 trials in function of $m$. We observe that, as for the linear case, there exists a minimal measurement number from which the signal direction $\bs x_0/\|\bs x_0\|$ is perfectly estimated, roughly from $m/s \geq 5$, with a success-failure transition around $m/s \simeq 3.4$. This is larger than the perfect reconstruction rate reached by CS from $m/s \geq 2.5$ with a transition point at $m/s \simeq 1.8$. However, this increase is meaningful. The CS reconstruction is achieved from $m$ complex measurements, equivalent to $2m$ real observations of $\bs x_0$, while the constraints brought by the model \eqref{eq:equiv-cs-model} correspond to a measurement model containing $m+2$ real observations (by considering the first two rows of $\bs A_{\bs z}$ associated with the normalization constraints). Therefore, $m$ observations have been lost from the phase-only sensing model. We can thus expect a ratio of (at least) 2 between the sample complexities of PO-CS and CS. This ratio could possibly be reduced by exploiting the unused the last constraint of \eqref{eq:consistency-real-imag}.      

In a second experiment, we have measured the robustness of PO-CS to the addition of a noise with bounded complex amplitude, as developed in Sec.~\ref{sec:noisy-phase-sensing}. We have considered the same experimental setup as above, corrupting the noisy PO-CS model \eqref{eq:PO-CS-noisy} with a uniform additive noise $\bs \epsilon \in \bb C^m$ with $\epsilon_k \sim_\iid \cl U( \tau \cl B)$ for $k \in [m]$, $\cl B = \{\lambda \in \bb C: |\lambda| \leq 1\}$. Since Rem.~\ref{rem:cond-tau} shows the existence of an additive noise with level greater than $\pi$ for which reconstructing the signal direction is impossible, we have set the noise level of this experiment to $\tau = \pi / 10^\alpha$ for $\alpha$ evenly sampled in $[1,3]$. The reconstruction procedure was performed using the basis pursuit denoising program $\Delta_{\cl K}(\bs A_{\bs z}, \bs e_1; \varepsilon)$ for $\varepsilon$ set\footnote{We postpone to a future work the definition of an accurate estimator of this quantity.} to the oracle value $\varepsilon = \|\bs A_{\bs \epsilon} \bs x\|$, and we solved \ref{eq:BP-def} with the SPGL1 python toolbox. As above, the results have been averaged over 100 trials for each association of the parameters $m$ and $\tau$, the quantities $\bs \Phi$, $\bs x$, and $\bs \epsilon$ being randomly regenerated at each trial.

In Fig.~\ref{fig:BP_rec_rate}(right), we plot the evolution of the average signal-to-noise ratio (SNR) $20 \log_{10} \|\bs x\|/\|\bs x - \hat{\bs x}\|$ (under the normalization hypothesis $\|\bs A \bs x\|_1 = \kappa \sqrt m$) as a function of $m/s$. For each sampled value of $\tau$, we display one average SNR curve colored according to the colorbar on the right. As expected, from $m/s \simeq 3$, the quality of the BPDN estimate improves, with an SNR level growing linearly with the increase of $-\log_{10} \tau/\pi$ (\ie for $\tau$ decaying logarithmically). This is involved by taking the logarithm of both sides of the instance optimality relation \eqref{eq:l2-l1-inst-opt} respected by BPDN.

\section{Conclusion and Perspectives}
\label{sec:conclusion}

With this work, we have shown that one can perfectly reconstruct the direction of a low-complexity signal belonging to a symmetric cone from the phase of its complex Gaussian random projections, \ie in the non-linear PO-CS model~\eqref{eq:PO-CS}. Inspired by~\cite{boufounos2013sparse}, this is achieved by formulating an equivalent linear sensing model defined from a signal normalization constraint, and a phase-consistency constraint with the observed signal phases in the random projection domain.  This reformulation allows us to perform this reconstruction with any instance optimal algorithm provided that the associated sensing matrix respects the RIP property over the considered low-complexity signal set. In addition, we inherit the stability and the robustness from such algorithms; the PO-CS model can be corrupted by a bounded noise with small noise level, and the observed signal is only required to be well approximated by an element of a low-complexity set (with small modeling error). We proved that a complex Gaussian random sensing matrix fulfills \whp the conditions allowing such a reconstruction; it leads to an equivalent sensing matrix (in the reformulated model) that respects \whp the restricted isometry property provided that the number of measurements is large compared to the complexity of the signal set, with an additional dependence to the noise level for noisy PO-CS.   

We can think of (at least) two open questions that are worth being investigated in future work. First, all the reconstruction guarantees involving the considered random matrices are non-uniform; they are valid \whp given the observed signal. However, up to an increase of the exponent of $1/\delta$ in the sample-complexity condition \eqref{eq:sample-complexity}, and in the case where the signal belongs to a low-complexity set $\cl K$ (zero modeling error), it should be possible to use the sign-product embedding~\cite{foucart2016flavors} and its extension to the complex field~\cite{FDVJ19} in order to ensure that, \whp, $\bs A_{\bs z}$ (with $\bs z = \signc(\bs A \bs x)$) is RIP for all~$\bs x \in \cl K$. 

Second, in several imaging applications (such as radar, magnetic resonance imaging, computed tomography), the complex sensing matrix representing the observational model is structured and amounts to (randomly) sub-sampling the rows of a Fourier matrix. In these contexts, phase-only compressive sensing provides an appealing sensing alternative, for instance, making the measurements insensitive to large amplitude variations and easing the measurement quantization. In correspondence with the inspirational work of Oppenheim~\cite{Oppenheim_1981,Oppenheim_1982}, a critical open question is thus to extend the RIP stated in Thm~\ref{thm:rip-for-Ax} and Thm~\ref{thm:noisy-rip} to partial random Fourier matrices or to other structured random sensing constructions. A first breakthrough in this direction would be to verify this extension for partial Gaussian circulant matrices, which satisfy \whp the $(\ell_1,\ell_2)$-RIP and are applicable to one-bit CS~\cite{Dirksen_2019}.

\end{document}